\documentclass[11pt]{article}

\usepackage{amssymb, amsmath, latexsym, amsfonts, graphicx, amsthm}
\usepackage[numbers]{natbib}
\usepackage[linesnumbered,ruled,vlined,boxed]{algorithm2e}
\usepackage{subcaption}
\usepackage[margin=1in]{geometry}
\usepackage{algorithmic}
\usepackage{epsfig}
\usepackage{graphicx}
\usepackage{graphics}
\usepackage{url}
\usepackage{verbatim}
\usepackage[margin=1in]{geometry}

\newtheorem{theorem}{Theorem}[section]
\newtheorem{definition}[theorem]{Definition}

\newtheorem{lemma}[theorem]{Lemma}

\newtheorem{example}[theorem]{Example}

\begin{document}
\title {\Large How to Manipulate Truthful Prior-dependent Mechanisms?}

\author{Pingzhong Tang\\
IIIS, Tsinghua University\\
Beijing, China\\
\texttt{kenshin@tsinghua.edu.cn}
\and Yulong Zeng\\
IIIS, Tsinghua University\\
Beijing, China\\
\texttt{cengyl13@mails.tsinghua.edu.cn}}

\date{}




	

\maketitle

\begin{abstract}

	In the standard formulation of mechanism design, a key assumption is
	that the designer has reliable information and technology to determine
	a prior distribution over types of the agents. In the meanwhile, as pointed out by the {\em Wilson's Principle},
	a mechanism should reply as little as possible on the accuracy of
	prior type distribution. In this paper, we put forward a simple model to formalize and justify this statement.
	
	In our model, each agent has a type distribution. In addition, the
	agent can commit to a fake type distribution and bids consistently
	with respect to the fake distribution (i.e., plays some Bayes equilibrium under
	the fake distributions). The model is partially motivated by a recent consensus in the literature that the prior distribution should not be endogenous to the model but has to be learned by the designer from the agents' past bids.
	
	We are interested in the equilibria of the induced distribution-committing games in several well-known mechanisms. Our results can be summarized as follows: (1) the game induced by Myerson auction under our model is strategically equivalent to the first-price auction under the standard model. Consequently, they are revenue-equivalent. (2) the second-price auction yields no less revenue than several reserve-based and virtual-value-based truthful, prior-dependent auctions, under our fake distribution model. These results complement a recent research agenda on prior-independent mechanism design.	
	
\end{abstract}
\newpage

\section{Introduction}

Many standard auction settings assume that buyers' prior type distributions are common knowledge among all buyers and the seller. For example, the seminal Myerson auction \cite{myerson1981optimal} uses this distribution knowledge to calculate the so-called {\em virtual values}. However, this assumption does not necessarily hold in practice: when facing new buyers, the seller may not know anything about the buyers' valuations. Even for repeated buyers, the seller still needs to infer the buyers' valuation distributions from the past bids by some, possibly inaccurate, algorithms. Such inaccuracy in type distributions may still compromise the theoretical optimality of such prior-dependent mechanisms.

\subsection{The model}

In this paper, we propose a model where buyers can credibly manipulate the prior distributions. In our model, each buyer has a true type distribution, just as in the standard mechanism design formulation. In addition, each buyer can pretend as if his type is from another distribution, which we call {\em the fake distribution}. The way that he convinces others, including the seller, that his type is indeed drawn from the fake distribution, is by consistently playing a Bayes Nash equilibrium in the mechanism {\em as if his type distribution is the fake distribution}. Such strategy is realistic as long as the buyers can commit to the fake distributions and doing so increases their expected payoffs.

The model captures a number of important ingredients in the practice of mechanism design. First of all, it
captures the fact that prior information is not endogenous to the
model, but has to come from the agents' past data. For example, in
sponsored search auctions \cite{edelman2005internet}, it is standard that the prior
distributions are estimated from the past bids \cite{nekipelov2015econometrics,ostrovsky2011reserve,chawla2014mechanism}. This makes the prior distributions a part of agents' strategy, rather than endogenous inputs to the mechanism. 

Secondly, it captures a certain dynamic aspect of mechanism design. Our model can be regarded as a sequence of iterations between 1) the seller adjusts their prior estimates of the buyers' valuations based on the bids observed so far; and 2) the buyers strategically  best respond, given the current prior estimate and the fact that their current bids will be used as future priors. Analysis of such dynamics can be complicated. Alternatively, our model focus on the analysis of the steady state resulted from a long, converged sequence of iterations. In our model, this steady state corresponds to a Nash equilibrium of the induced distribution-reporting game.

Thirdly, it helps explain certain underbidding
behaviors in prior-dependent auctions. Many truthful auctions in standard setting are no longer ``truthful" in fake distribution setting.  

Last but not least, the model allows one to
quantify the damage caused by relying on prior information and to
prescribe more realistically the outcomes one should expect in these auctions.

\begin{example}

Consider the Myerson auction where there are two buyers, each with valuation drawn independently identically from the uniform distribution on $[0,1]$. Assume that the first buyer always bids his true valuation. In other words, the buyer's bid distribution is also a uniform distribution on $[0,1]$. We are interested in what the second buyer should behave. Suppose that the second buyer always bids a small positive number $\epsilon$. By observing the bids distributions for a sufficiently long period, the seller would believe that the first buyer's valuation from the uniform $[0,1]$, while the second buyer valuation is exactly $\epsilon$.

In the Myerson auction,the first buyer's virtual value is $2v_1-1$ when his true type is $v_1$, while the second buyer's virtual value is always $\epsilon$. The item is allocated to the second buyer if and only if $2v_1-1<\epsilon$, of which the probability is $(\epsilon+1)/2$. As a result, the second buyer obtains an expected utility of $\int_0^1 \frac{\epsilon+1}{2} (v-\epsilon) dv=\frac{(\epsilon+1)(1-2\epsilon)}{4}$. 

Compared to the case where the second buyer commits to the truthful distribution (i.e., uniform $[0,1]$) and obtains an expected utility of $1/12$, his utility strictly increases. We will later prove that each buyer reporting a uniform distribution on $[1/4,1/2]$ is an equilibrium in the induced (distribution-reporting) game.

\end{example}

One might argue that the model assumes too much rationality on the buyer side, since the buyers need to figure out the equilibrium of fake distributions. One justification is that, the agents may not need to compute the correct fake distribution upfront, but can learn to do this over time via some learning dynamic that converges to the correct distributions in the steady state as mentioned.


%

\subsection{Our contributions}

We report a few interesting findings when applying this model to the single-item auction domain. To justify why we can focus
on truthful auctions, we first prove a version of revelation principle for this
model, stating that one can without loss of generality restrict attention to truthful auctions.

\noindent {\bf Theorem 1.1} For any prior-independent mechanism family, it is a weakly dominant strategy for each buyer to report his true distribution.

This suggests that it might be more interesting to look at prior-dependent auctions. We then apply this model to the revenue-optimal auction, aka. the
Myerson auction, where we prove one of our main theorems that says the Myerson auction under  fake distribution setting
is strategically equivalent to the first-price auction in the
standard setting.

\noindent {\bf Theorem 1.2}. Given buyers' true distributions, the induced distribution-reporting game of Myerson auction is strategically equivalent to first-price auction in the standard setting.

As a result, the two games are revenue equivalent. In other words, Myerson revenue in the fake distribution model is reduced to the revenue of the first-price auction game in the standard setting.

We extend our analysis to a similar prior-dependent auction that has been studied in the literature, called second-price auction with monopoly (i.e., Myerson's)
reserve~\cite{hartline2009simple,dhangwatnotai2015revenue}.

\noindent{\bf Theorem 1.3} Given that the buyers' true distributions are i.i.d. regular, the revenue of second-price auction with monopoly reserve under the equilibrium of fake distribution is the same as the revenue of second-price auction under the true (or equivalently, fake) distribution.


It is important to distinguish Myerson auction and second-price auction with monopoly reserve, even in the i.i.d. case. The allocation rules of the two auctions are the same under the true distribution, but different when computing the equilibrium of fake distributions: to consider a possible deviation, a buyer can report any fake distribution, creating an instance of non-i.i.d. distributions, where the allocation rules for both auctions differ.

We further generalize the idea of prior-dependent reserve pricing by analyzing second
price auction with a random reserve, which has also been extensively studied in the literature \cite{huang2015making,dhangwatnotai2015revenue}.

\noindent{\bf Theroem 1.4} Given that  the buyers' true distributions are i.i.d., the revenue of second-price auction with any single-sample (aka. random) reserve
under the equilibrium of fake distribution is no more than the revenue of second-price auction under the true distribution.

We finally apply this model to a general class of truthful auctions called
virtual efficient auctions, in which the allocation rule is efficient,
with respect to a general definition of virtual value function.
\noindent{\bf Theroem 1.5}  Given that the buyers' true distributions are i.i.d., the revenue of any virtual efficient mechanism family under the equilibrium of fake distribution is no more than the revenue of second-price auction under the true distribution.

To sum up, for all symmetric cases, second-price auction
is superior compared to other auctions considered in this paper. The
comparison between Myerson and second-price auction is unknown in
the asymmetric case since the BNE of First-price auction is
unknown \cite{chawla2013auctions,syrgkanis2013composable}.

%
%

\subsection{Relative work}

Recently, there has been a large body of literature on the problem of
prior-independent mechanism design \cite{bergemann2011should,bergemann2008pricing,devanur2011prior,fu2013prior,fu2015randomization}. One of the major
motivations of these papers is the so-called Wilson's principle \cite{wilson1985game}
, which states that a mechanism should reply as little as
possible on the prior distributions. Over the years, the principle has
guided many practical auction designs and has been
consistent with experiences: auctions that heavily rely on prior
distributions are in general not robust and rarely used. Our model can be regarded as a quantitative model for the Wilson's principle.

Independent of our work, a few recent papers~\cite{mohri2015revenue,chen2016bayesian,devanur2015perfect} also consider the general problem where a {\em single} strategic buyer repeatedly interacts with the seller.  The focus of all these papers is on the dynamic and learning aspect of the problem where the buyer learns to play equilibrium. Mohri and Medina~\cite{mohri2015revenue} consider a model where the seller posts a price in each time period to a buyer whose valuation is independent random draw in each period. They introduce a notation called  $\epsilon-strategic$ buyer where the buyer does not necessarily best respond in each round. They design a sequential posted prices that achieves additive bounds to optimal revenue in two different settings. Chen and Wang \cite{chen2016bayesian} study the same setting and make the assumption that the action of the buyer in each round is so chosen as to maximize his overall discounted utility. They show that the seller can design a posted price mechanism that can almost surely learn the buyer's true type and achieve constant additive bound to optimal revenue. Devanur et. al.~\cite{devanur2015perfect} analyze the perfect Bayesian equilibria in a repeated game where the buyer acts strategically and the seller updates his belief of the prior distribution. We remark that our model, even though admits a dynamic and learning interpretation, is still a single-shot distribution-reporting game. 

Caillaud and Robert~\cite{caillaud2005implementation} consider implementation of Myerson auction in a prior-independent fashion. They assume that at least one buyer knows the prior distribution of the winning buyer. Izmalkov~\cite{izmalkov2004shill} also discusses the implementation of the revenue maximization auction via English auction by shill biddings. While both work and ours are related in that we all study mechanism design from the perspective of priors, our goal is different in that we focus on the study of strategical manipulation of prior distributions.



\section{The fake distribution model}
In a single item auction, an auctioneer has a single item to sell to $n$ buyers. Each buyer $i$ has a private valuation $v_i$ towards the item, where each $v_i$ is drawn from a distribution $F_i$. Bulow and Roberts \cite{bulow1989simple} give a convenient representation of a distribution called \emph{quantile}.
\begin{definition}(Quantile).
	Given a random variable $v$ with cumulative distribution $F(v)$, let $q=1-F(v) \in [0,1]$ denote the quantile of $v$. The function $v(q)=F^{-1}(1-q): [0,1]\rightarrow \mathbb{R}_{+} $ is a mapping from the quantile space to valuation space, and is uniquely determined by distribution $F$. 
\end{definition}

We use $v(\cdot)$ to denote the distribution of $v$ in the quantile space. Note that $q$ is always uniformly distributed on [0,1] and $v(\cdot)$ is weakly monotone decreasing. Let $R(q)=qv(q)$ denote the {\em revenue curve} of distribution $v(\cdot)$.

We write $\mathbf{q}=(q_1,q_2,...,q_n)$ for the profile of quantiles for all bidders. Let $\mathcal{F}$ denote the set of all weakly monotone decreasing functions that maps $[0,1]$ to $\mathbb{R}_{+}$. We write $\mathbf{f}(\cdot)=(f_1(\cdot),f_2(\cdot),...,f_n(\cdot)) \in \mathcal{F}^n$ for the profile of $n$ arbitrary distributions in $\mathcal{F}^n$. Now we suppose that each $v_i$ is drawn from distribution $v_i(\cdot) \in \mathcal{F}$. We write $\mathbf{v}(\cdot)=(v_1(\cdot),...,v_n(\cdot))$ for the profile of valuation distributions. Sometimes we call this profile the {\em true distributions}. In our setting, the true distributions $\mathbf{v}(\cdot)$ is unknown to the seller\footnote{To compute equilibrium distribution of the induced game in the steady state, a sufficient condition is to assume that $\mathbf{v}(\cdot)$ is common knowledge among the buyers.}.


In our model, each buyer reports a new distribution $\hat{v}_i(\cdot)$ that may be different from $v_i(q)$, called the {\em fake distribution}.  We use $\hat{\mathbf{v}}(\cdot)=(\hat{v}_1(\cdot),...,\hat{v}_n(\cdot))$ to denote the profile of fake distributions.

\begin{definition}(Mechanism)
	
	Given prior $\mathbf{f}(\cdot)$, a (direct) mechanism consists of an allocation rule $x_i:\mathbb{R}^n_{+} \rightarrow [0,1]$ and a payment rule $t_i:\mathbb{R}^n_{+} \rightarrow \mathbb{R}$, for $i =1,2,...,n$.
\end{definition}

Note that both $x_i$ and $t_i$ are implicitly dependent on the prior. This motivates the following definition of mechanism family, which explicitly addresses this fact.

\begin{definition}(Mechanism Family)
	
	A mechanism family $\{M^{\mathbf{f}(\cdot)}|\mathbf{f}(\cdot) \in \mathcal{F}^n\}$ is a collection of mechanisms, one for each prior, so that each mechanism $M^{\mathbf{f}(\cdot)}$ in it takes $\mathbf{f}(\cdot)$ as the prior.
	
\end{definition}


Under a mechanism, we use function $b_i(\cdot): [0,1] \rightarrow \mathbb{R}_{+}$ to denote bidder $i$'s bidding strategy.

Given prior $\mathbf{f}(\cdot)$, bidding strategy profile $\mathbf{b}(\cdot)=(b_1(\cdot),...,b_n(\cdot))$ and buyer $i$'s quantile $q_i$, define the {\em interim allocation} as follows:

$$x^*_i(q_i)=\mathbf{E}_{b_{-i}}[x_i(b_i,b_{-i})]= \int_{\mathbf{q}_{-i}} x_i(b_i(q_i),\mathbf{b}_{-i}(\mathbf{q}_{-i})) d\mathbf{q}_{-i},$$

and {\em interim payment} as follows:

$$t^*_i(q_i)= \mathbf{E}_{b_{-i}}[t_i(b_i,b_{-i})]= \int_{q_{-i}} t_i(b_i(q_i),\mathbf{b}_{-i}(\mathbf{q}_{-i})) d\mathbf{q}_{-i}.$$




In order for the seller and other bidders to verify the fake distribution, the bidders need to play a Bayes Nash equilibrium of the mechanism, under the fake distribution profile. Formally


\begin{definition}(Bayes Nash equilibrium (BNE))
	Given $M^{\hat{\mathbf{v}}(\cdot)}$, the bidding strategy profile $\mathbf{b}(\cdot)=(b_1(\cdot),b_2(\cdot),...,b_n(\cdot))$ is a BNE of the mechanism, if
	\begin{eqnarray*}
		b_i(q_i) &\in& \arg\max_{b} \int_{q_{-i}} (\hat{v}_i(q_i)x_i(b,b_{-i}(q_{-i})) -t_i(b,b_{-i}(q_{-i})) )dq_{-i} \\
		&=&\arg\max_{b_i(q_i)} \hat{v}_i(q_i)x^*_i(q_i)-t^*_i(q_i) , ~\forall i=1,...,n ~\forall q_i \in [0,1]
	\end{eqnarray*}
\end{definition}

If the bidding strategy profile $\mathbf{b}(\cdot)=\hat{\mathbf{v}}(\cdot)$ is a BNE, we say the mechanism $M^{\hat{\mathbf{v}}(q)}$ is Bayes Incentive Compatible (BIC). If $\hat{v}_i(q_i)x^*_i(q_i)-t^*_i(q_i)\geq 0, ~\forall i=1,...,n~\forall q_i \in [0,1]$, we say it is {\em ex- post} individually rational (IR).

With the notations ready, the timing of our model can be conveniently described as follows:

\begin{enumerate}
	\item The seller commits to a mechanism family $\{M^{\mathbf{f}(\cdot)}|\mathbf{f}(\cdot) \in \mathcal{F}^n\}$.
	
	\item Buyers choose fake distributions $\hat{\mathbf{v}}(\cdot)$.
	
	\item Buyers play BNE of the mechanism $M^{\hat{\mathbf{v}}(\cdot)}$.
	
	
\end{enumerate}
Let $\Phi$ denote the set of all one-on-one mappings $\phi: [0,1] \rightarrow [0,1]$.

We, as analysts, are interested in the equilibrium of the induced game:

\begin{definition}(Induced game)
	\label{induced_game}
	Given a mechanism family $\{M^{\mathbf{f}(\cdot)}|\mathbf{f}(\cdot) \in \mathcal{F}^n\}$), the induce game is a normal-form game $(N,A,U)$ where
	\begin{itemize}
		\item $N=\{1,...,n\}$ is the set of buyers,
		\item $A=A_1 \times ...\times A_n$. where $A_i =\mathcal{F} \times \Phi$ is the set of actions of buyer $i$. We use $(\hat{v}_i(\cdot)\in \mathcal{F}, \phi_i(\cdot) \in \Phi)$ to denote an action of buyer $i$. This means that, the buyer not only needs to choose a fake distribution $\hat{v}_i(\cdot)$, but also needs to explicitly use $\phi_i(\cdot)$ to map the true distribution to the fake distribution, from quantile to quantile.
		\item $U=\{U_1,...,U_n\}$: the utility function for each buyer.
	\end{itemize}	
\end{definition}

The necessity of $\phi_i(\cdot)$ follows from the definition of utility of the induced game.

In the definition above, $U_i: A \rightarrow \mathbb{R}$ denotes the utility function of buyer $i$, such that
$$U_i((\hat{v}_1(\cdot),\phi_1(\cdot)),...,(\hat{v}_n(\cdot),\phi_n(\cdot)))= \int_0^1 (x_i^*(q_i)v_i(\phi_i(q_i))-t^*_i(q_i))dq_i$$
Where $x_i^*(q_i)$ and $t^*_i(q_i)$ are the interim allocation and payment in mechanism $M^{\hat{\mathbf{v}}(\cdot)}$ at BNE.


After buyer $i$ reports a fake distribution $\hat{v}_i(\cdot)$, by the definition of $\phi_i(\cdot)$, buyer $i$ will bid as if his valuation is $\hat{v}_i(q)$, when his true valuation is $v_i(\phi(q))$. We will prove in next section that each buyer $i$ will always want to choose $\phi_i(q_i)=q_i, \forall q_i \in [0,1]$ to maximize his utility. As a result, we can omit $\phi$ from the definition of the induced game and equivalently write $A_i =\mathcal{F}$ for short. 

Given $\hat{\mathbf{v}}(\cdot)$,
$$U_i(\hat{\mathbf{v}}(\cdot))= \int_0^1 (x_i^*(q_i)v_i(q_i)-t^*_i(q_i))dq_i$$

Denote by $u^*(q_i)=x_i^*(q_i)v_i(q_i)-t^*_i(q_i)$ the {\em interim utility} for any fixed $q_i$, so that $$U_i(\hat{\mathbf{v}}(\cdot))=\int_0^1 u^*(q_i) dq_i.$$

Game theoretical notions, such as Nash Equilibrium and best response are standard in the induced game. We sometimes refer to Nash equilibrium simply as equilibrium in the induced game.

\begin{definition}
	Given $\hat{\mathbf{v}}(\cdot)$, the seller's revenue is
	\begin{eqnarray*}
		REV&=&  \sum_{i=1}^n  \int_0^1 t^*_i(q_i) dq_i= \sum_{i=1}^n \int_{\mathbf{q}} t_i(\mathbf{b}(\mathbf{q})) d\mathbf{q}
	\end{eqnarray*}
	
\end{definition}

\section{Technical preliminaries}

In this section, we present several technical lemmas useful for the analyses of equilibrium in the induced game of various mechanism families. As mentioned, we only focus on the truthful mechanism families. This is justified by the following {\em revelation principle for mechanism families}. It states that, it is with loss of generality to restrict attention to mechanism families where every mechanism in it is Bayesian incentive compatible (BIC).


\begin{lemma}(Revelation principle for mechanism families)
	For any mechanism family $\{ M^{\mathbf{f}(\cdot)}|\mathbf{f}(\cdot) \in \mathcal{F}^n\}$, and given a BNE for each $M^{\mathbf{f}(\cdot)}$, there exists a new mechanism family $\{M'^{\mathbf{f}(\cdot)}|\mathbf{f}(\cdot) \in \mathcal{F}^n\}$, such that each mechanism $M'^{\mathbf{f}(\cdot)}$ is BIC and implements the same outcome as $M^\mathbf{f}(\cdot)$ at its truthful BNE. Moreover, if $\hat{\mathbf{v}}(\cdot)$ is a Nash equilibrium of the induced game of $M$, it is also a Nash equilibrium of the induced game of $M'$
	
\end{lemma}

A mechanism family is truthful if every mechanism in it is BIC.

With this lemma, we only consider truthful mechanism family $\{M^{\mathbf{f}(\cdot)}|\mathbf{f}(\cdot) \in \mathcal{F}^n\}$ and the truthful BNE of every mechanism in it.

In the following, we use the interim notations and analyze from the point of view of buyer $i$ as if he is the only buyer. From now on, we omit the subscript $i$ when it is clear from the context.

\begin{lemma}
\label{mapping}
	For any truthful mechanism family and fixed $\hat{\mathbf{v}}(\cdot)$, the mapping $\phi(q)=q,\forall q \in [0,1]$ maximizes the buyer's utility among all possible one-on-one mappings.
\end{lemma}
\begin{proof}
All missing proofs can be found in the appendix.	
\end{proof}
In other words, in the induced game, the buyer will always want to map a value that is in quantile $q$ in the true distribution to a value that is also in quantile $q$ in the fake distribution, no matter what the fake distribution is.

This lemma also holds for non-truthful mechanism family, since weakly decreasing $x^*(q)$ is a necessary condition of a BNE.


Use the same approach as Myerson's analysis, we first prove the {\em payment identity} formula using the quantile representation.

\begin{lemma}(Payment identity)
\label{payment_identity}
	In a BIC mechanism $M^{\hat{\mathbf{v}}(\cdot)}$, suppose $\hat{v}(\cdot)$ is continuous, then the interim allocation $x^*(q)$ is weakly decreasing in $q$. Moreover, given the buyer's quantile $q$,  the interim payment satisfies
	
	\begin{equation}
	\label{pay_iden}
	t^*(q)= \hat{v}(q)x^*(q)+\int_{r=q}^1 x^*(r)\hat{v}'(r) dr-\hat{v}(1)x^*(1)+t^*(1) \\
	\end{equation}
	
\end{lemma}
Given allocation rule $\textbf{x}(\cdot) \in \mathbb{R}_+^n \rightarrow [0,1]^n$ such that $x^*(q)$ is weakly decreasing, the auctioneer can implement (\ref{pay_iden}) by choosing payment rule
\begin{equation}
\label{payment}
t(\mathbf{b}(\mathbf{q})) = b(q)x(\mathbf{b}(\mathbf{q}))-\int_{z=\hat{v}(1)}^{b(q)} x(z,\mathbf{b}_{-i}(\mathbf{q}_{-i}))dz, ~~\forall \mathbf{q} \in [0,1]^n
\end{equation}
Payment rule (\ref{payment}) satisfies $-\hat{v}(1)x^*(1)+t^*(1)=0$

Lemma \ref{payment_identity} requires the continuity and boundedness of $\hat{v}(\cdot)$, In fact, we show in appendix that it is sufficient to only consider continuous, bounded fake distributions. So WLOG, we can make this assumption throughout the paper.


\begin{lemma}
	\label{utilitity1}
	Given fake distribution profile $\hat{\mathbf{v}}(\cdot)$, the expected utility of the buyer is
	\begin{equation}
	U= \int_0^1 x^*(q)(v(q)-r(q)) dq+\hat{v}(1)x^*(1)-t^*(1)\\
	\end{equation}
	where $r(q)=\hat{v}(q)+q\hat{v}'(q)$.
\end{lemma}


$r(q)$ is known as the virtual value (in \cite{myerson1981optimal}).


\begin{lemma}
	\label{revenue}
	Given the fake distribution profile $\hat{\mathbf{v}}(\cdot)$, the expected revenue is
	\begin{equation}
	REV= \sum_{i=1}^n (\int_0^1 x_i^*(q)r(q)dq -\hat{v}_i(1)x_i^*(1)+t_i^*(1))
	\end{equation}
\end{lemma}

Lemma 3.5 follows directly from Lemma 3.4 and Definition 2.6.

\section{Warm-up: Prior-independent mechanism family}

We first analyze the easy case where the mechanism family under consideration is independent of priors. We show that in any such mechanism family, it is a weak dominant strategy for each buyer to report the true distribution.

\begin{definition}(Prior-independent mechanism family)
	
	A mechanism family $\{ M^{\mathbf{f}(\cdot)}|\mathbf{f}(\cdot) \in \mathcal{F}^n\}$ is prior independent if it is truthful, and take any bid profile $\mathbf{b}=({b}_1,{b}_2,...,{b}_n)$ as input, $M^{\mathbf{f}(\cdot)}$ returns the same outcome for any $\mathbf{f}(\cdot)$.
\end{definition}

By definition, second-price auction (SPA) is a prior-independent mechanism family, while the Myerson auction, second-price auction with monopoly reserve price, are not.

\begin{theorem}
	\label{prior-independent}
	For any prior-independent mechanism family, it is a weakly dominant strategy for each buyer to report the true distribution in the induced game. Moreover, the buyer's utility increases as $\hat{v}(q)$ approaches $v(q)$.
\end{theorem}



	


\section{The Myerson auction}

In the fake distribution setting, we show that the Myerson auction is strategically equivalent to the first-price auction in the standard setting.

We first consider the case where we restrict the set of possible fake distributions to regular distributions, i.e., the virtual value $r(q)=\hat{v}(q)+q\hat{v}'(q)$ is weakly decreasing.

For each $i$, given buyer $i$'s fake distribution $\hat{v}_i(\cdot)$ and bid $b_i$, define virtual bid (of buyer $i$) $ \phi(b_i)=b_i-\frac{1-G_{i}(b_i)}{g(b_i)}$, where $G_i$ is the cumulative distribution function on $\hat{v}_i(\cdot)$ and $g_i$ is the probability density function.

\begin{definition}
	A mechanism $M^{\mathbf{f}(\cdot)}$ in the Myerson mechanism family allocates the item to the buyer with the largest positive virtual bid $\phi(b)$ (with respect to $f(\cdot)$). (If no buyer has positive virtual bid, the seller keeps the item). The payment rule is according to (\ref{payment}).

\end{definition}

Under any regular fake distribution profile, it is a truthful mechanism family.

Fix other buyers' fake distributions $\hat{v}_{-i}(\cdot)$,  given the quantile $q$, in the same spirit as the proof of Theorem \ref{prior-independent},  the interim allocation only depends on the value $r(q)$.  As a result, we can write $x(r(q))=x^*(q)$.

\begin{theorem}
	In Myerson mechanism family, for fixed $\hat{v}_{-i}(\cdot)$, one of the buyer's best response of the induced game is
	$$\hat{v}(q)=\frac{\int_0^q r(s)ds}{q},q \in [0,1]$$
	\begin{equation}
	r(q) = \arg\max_{r(q)\geq 0} x(r(q))(v(q)-r(q)) , q \in [0,1]
	\end{equation}
	Furthermore, $\hat{v}(q)$ is weakly monotone decreasing and regular.
\end{theorem}

\begin{example}
	\label{Myerson_example1}
	For the case of $n=2$ and both buyers' true distributions are i.i.d. from uniform [0,1], i.e., $v_1(q)=v_2(q)=1-q$. Suppose that buyer two reports the true distribution, i.e. $\hat{v}_2(q)=1-q$, so $r_2(q)=1-2q$. The best response of buyer one in the induced game is $\hat{v}_1(q)=\epsilon$.
	
\end{example}
In this example, buyer one's utility is $\frac{1}{4}$, greater than $\frac{1}{12}$, the utility buyer one gets when reporting the true distribution.

\begin{example}
	\label{example_2_buyer}
	For the case of $n=2$ and both buyers' true distributions are i.i.d. from uniform [0,1], i.e., $v_1(q)=v_2(q)=1-q$. Each buyer reports fake distribution $\hat{v}(q)=\frac{1}{2}-\frac{1}{4}q$ is an equilibrium of induced game in Myerson's mechanism family.
\end{example}

For the equilibrium in Example \ref{example_2_buyer}, $U_1=U_2=\frac{1}{6}$, expected social welfare \\
$SW=E[\max\{v_1(q),v_2(q)\}]=\frac{2}{3}$, revenue $REV=\frac{1}{3}$.

Compared to the Myerson auction in the standard setting: $U_1=U_2=\frac{1}{12}, REV=\frac{5}{12},SW=\frac{7}{12}$, the equilibrium in the induced game yields higher buyer utility and social welfare, but less revenue. As a result, it is desirable for the buyers to commit to such fake distributions.\\

We now prove our main theorem of the section, the equivalence between Myerson auction in the fake distribution model and first-price auction in the standard model.

\begin{theorem} {\bf (Main Theorem)}
	\label{equal_to_FPA}
	Given $\mathbf{v}(\cdot)$, the induced game of Myerson mechanism family is strategically equivalent to the first-price auction in standard setting.
\end{theorem}

The idea for proving the theorem is to notice that, the utility
formula of the induced distribution-reporting game of Myerson (see Equation (\ref{F}) in the appendix) is
exactly the utility formula of a first-price auction where each
player's action is to report a virtual value.


As a result, the two games are revenue equivalent. In other words, Myerson's revenue in the fake distribution setting is reduced to that of the first-price auction in the standard setting.


%
%

See again  Example \ref{example_2_buyer}, it is well known that the BNE for first-price auction is $b(v)=\frac{v}{2}$ for valuation distribution uniform in [0,1]. So the virtual value of fake distribution $r(q)=\frac{v(q)}{2}=\frac{1-q}{2}$, same as the proof of Example \ref{example_2_buyer}.

Given this result, it is easy to prove the monotonicity of $r(q)$. Since $r(q)$ is the optimal bidding strategy in the corresponding FPA, it is monotone.

All results so far extend to the case where we allow the buyers to report any, possibly irregular fake distribution. The Myerson auction also extends the general distributions via ironing, a procedure that maps the irregular distribution to a regular distribution, then run Myerson mechanism for the ironed distributions. All our results in this section remain the same, because whenever a buyer reports an irregular distribution, it is equivalent to report the ironed distribution instead.

\section{Second-price auction with monopoly reserve}

Second-price auction with monopoly reserve \cite{hartline2009simple} is an auction that first rejects all the bids below the respective monopoly reserves and run second-price auction on the remaining bids. Since we allow buyers to report any distributions, we need to extend the definition of monopoly reserve to the irregular case.  We begin by reviewing Myerson's ironing procedure.

\begin{definition}(Ironing)
	
\noindent For any distribution $f(\cdot) \in \mathcal{F}$ (regular or irregular) with revenue curve $R(q)=qf(q)$. Let $\overline{R}(q)$ denote the smallest concave function such that $\overline{R}(q) \geq R(q)$ for all $q\in [0,1]$. $\overline{R}(q)$ is called the ironed revenue curve. The derivative of ironed revenue curve $\overline{r}(q)=\overline{R}'(q)$ is the ironed virtual value.
	
\end{definition}

Given a distribution $f(\cdot)$, the quantile $q^*$ such that $\overline{r}(q^*)=0$ is called the reserve quantile for distribution $f(q)$. The value $f(q^*)$ is called the monopoly reserve price for distribution $f(q)$. If there are multiple reserve quantiles, set the reserve quantile to be the largest one (i.e., the smallest type whose virtual value is 0).
\begin{definition}(Second-price auction with monopoly reserve (SPAMR) mechanism family)
	
	For any mechainism $M^{\mathbf{f}(\cdot)}$ in a SPAMR
	mechanism family, it first computes the monopoly reserve price for each buyer $i$ and then
	allocates the item to the buyer with the largest bid $b$ which is not less than the monopoly reserve price. (If no buyer satisfies this, the seller
	keeps the items). The payment rule is according to (\ref{payment}).
\end{definition}




\begin{lemma}
	\label{ironing}
	For any distribution $f(q)$ with revenue curve $R(q)=qf(q)$, its reserve quantile $q^*$ satisfies $q^*=\arg\max_{q\in [0,1]} qR(q)$ (breaking tie by choosing the largest quantile)
\end{lemma}

\begin{proof}
	Since $\overline{R}(q)$ is convex hull of $R(q)$, their maximum points remain the same. So $\arg\max \overline{R}(q)=\arg\max R(q)$, and by definition $q^*=\arg\max{\overline{R}(q)}=\arg\max R(q)$
\end{proof}



Unlike Theorem 4.2, reporting truthfully now leads to a potentially bad reserve price and thus may no longer be a dominant strategy.

In this section, we assume that the true distributions $\mathbf{v}(\cdot)$ are regular. We further consider two cases: in this first case, the mechanism family allows buyers to report any fake distributions, while in the second case, the mechanism family only allows buyers to report regular fake distributions.

\subsection{Case one: general fake distribution}

\begin{theorem}
	\label{vqformall}
	In the induced game of SPAMR mechanism family, fix $\hat{v}_{-i}(\cdot)$, there exists $R$ such that, it is WLOG to consider $\hat{v}(\cdot)$ with the following form (see  the red line (line $Oq_1q_2q_3\hat{q}^*$) in Figure 1 in Appendix).
	\begin{itemize}
		\item for $q_3 \leq q \leq 1$, $\hat{v}(q)=R$
		\item for $q_2 \leq q < q_3$ or $0 \leq q < q_1$, $\hat{v}(q)=v(q)$
		\item for $q_1 \leq q < q_2$, $\hat{v}(q)=\frac{R}{q}$
	\end{itemize}
	Where $q_1v(q_1)=R,q_2v(q_2)=R,q_1<q_2,v(q_3)=R$.
	
\end{theorem}

\begin{theorem}
\label{no_equilibrium}
	For any buyers' true distributions that are i.i.d. regular with minimum support at 0 and do not have a point mass, there does not exist a symmetric equilibrium in the induced game.
\end{theorem}

\subsection{Case two: regular fake distribution}


\begin{theorem}
	\label{rvqform}
	There exists a quantile $q_0$ such that it is without generality to consider $\hat{v}(\cdot)$ with the following form
	\begin{itemize}
		\item For all $q<q_0$, $\hat{v}(q)=v(q)$
		\item For all $q\geq q_0$, $\hat{v}(q)=\frac{q_0v(q_0)}{q}$
	\end{itemize}
\end{theorem}
We call the corresponding $q_0$ the incident quantile of the buyer.
\begin{example}
	\label{example_SPAMR}
	In SPAMR mechanism family, for the case where $n=2$ and the buyers' true distribution are i.i.d. from uniform [0,1], i.e. $v_1(q)=v_2(q)=1-q$, then the following fake distribution is an equilibrium of the induced game.
	
	\begin{itemize}
		\item for $0 \leq q< \frac{1}{4}$ ,$\hat{v}_i(q)=1-q,$
		\item for $\frac{1}{4} \leq q \leq 1$, $\hat{v}_i(q)= \frac{3}{16q}$
	\end{itemize}
\end{example}
\begin{theorem}
	\label{revenue_equal}
	In SPAMR family, for the case where the buyers' true distributions are i.i.d. regular, the revenue yielded by the symmetric Nash equilibrium of the induced game is equal to the revenue of the SPA under the true distribution.
\end{theorem}

\section{Second-price auction with reserves in quantile}

\begin{definition}{Second-price auction with random quantile reserve (SPARQR)}
	
	Each mechanism $M^{\mathbf{f}(\cdot)}$ in the mechanism family first a random drawn $q_r$ from uniform $ [0,1]$, and allocate the item to the buyer with
	largest bid $b$ subject to $b \geq f(q_r)$ (If none of the buyers satisfies the constraint, the
	seller keeps the items). The payment rule is derived by (\ref{payment}).
	
\end{definition}

If the distribution $f(\cdot)$ does not have a point mass, the mechanism is known as the {\em single sample mechanism}~\cite{huang2015making}.

\begin{theorem}
	\label{ssvqform}
	In SPARQR mechanism family, for the case where $n=2$ and buyers' true distributions are i.i.d., fake distribution $\hat{v}(\cdot)$ subject to $v(q)=\hat{v}(q)-q\hat{v}'(q),~\forall q \in [0,1]$ and $v(1)=\hat{v}(1)$ is weakly decreasing and is an equilibrium distribution of the induced game.
\end{theorem}

\begin{example}
	In SPARQR mechanism family, for the case where  $n=2$ and buyers' true distributions are i.i.d. from uniform [0,1], then $\hat{v}(q)=1-q+q\ln q$ is an equilibrium of the induced game.
	
	Note that $\ln q>\frac{q-1}{q}$ so $\hat{v}(q) \geq 0$,
	$\hat{v}'(q)=\ln q \leq 0$, which is a feasible distribution.
\end{example}

\begin{theorem}
	\label{singel_sample_less_than_SPA}
	In the SPARQR mechanism family, if $n$ buyers' true distributions are i.i.d., the revenue under the equilibrium of the induced game is no more than the revenue of second price auction under the true distribution.
\end{theorem}

\section{Virtual efficient mechanism family}

Most auction family in the single-item case can be regarded as a comparison of specific ``virtual value'' \cite{cai2012algorithmic}, in which each mechanism allocates the item to the buyers with the largest ``virtual value'' and the payments are set according to (\ref{payment}). Here, the ``virtual value'' is defined with respect to a buyer's fake distribution. In this section, we consider a class of truthful mechanism family in which the ``virtual value'' $R$ can be written as a simple functional of the fake distribution $\hat{v}(q)$, i.e., $R=R(q,\hat{v}(q),\hat{v}'(q))$.

For the buyer with fake distribution $\hat{v}(q)$, in the value space, we define $R$ equivalently, i.e., $R(v)=R(v,G,g)$, where $G$ is the cumulative distribution function on $\hat{v}(q)$ and $g$ is the probability density function. Given a bid $b$, let  $R(b)$ denote the ``virtual bid'' .

\begin{definition}
	In a virtual efficient mechanism family (VE), each mechanism allocate the item to the buyers with the largest ``virtual bid'' $R(b)$.
	The payment rule is set according to (\ref{payment}).
\end{definition}

Note the Myerson mechanism and second-price with monopoly reserve are not included in this definition, since both auctions use the ironed virtual value $\overline{r}(q)$, which can not be written as a simple functional for irregular fake distributions.

\begin{lemma}
	\label{coefficients}
	For a truthful VE mechanism family,  $\frac{\partial R}{\partial \hat{v}'(q)}=0$, $\frac{\partial R}{\partial q} \leq0$, $\frac{\partial R}{\partial \hat{v}(q)} \geq 0$.
\end{lemma}

\begin{lemma}
	\label{virtual_efficient}
	For a truthful, VE mechanism family, in the i.i.d. case, if there exists a symmetric equilibrium, then the fake distribution $\hat{v}(\cdot)$ in equilibrium satisfies $\hat{v}(q) \leq v(q)$ for any $q$.
	
\end{lemma}

\begin{theorem}
	\label{less_than_second_price}
	If the buyers' true distributions are i.i.d., the revenue of any
	VE family under the equilibrium of induced game is no
	more than the revenue of second-price auction under the true distributions.
\end{theorem}

\bibliographystyle{plain}
\bibliography{fake_distribution}

\clearpage
\appendix
\section*{APPENDIX}
\setcounter{section}{0}

\section{Missing proofs from section 3}

\begin{proof}[Proof of Lemma \ref{mapping}]
		By definition, the buyer's utility in the induced game is $U=\int_0^1 (x^*(q)v(\phi(q))-t^*(q))dq$. When the buyers' fake distributions are given, $x^*(q)$ and $t^*(q)$ are constants.

		By the property of BIC, we have that $x^*(q)$ is decreasing in $q$. Since $\phi(q)$ is a mapping over $[0,1]$, by the Rearrangement Inequality, we have
		$$\int_0^1 x^*(q)v(\phi(q))dq \leq \int_0^1 x^*(q)v(q)dq$$
		
		When the mapping is $\phi(q)=q$, the utility is
		$U'=\int_0^1 (x^*(q)v(q)-t^*(q))dq$, so $U \leq U'$
		
\end{proof}
\begin{proof}[Proof of Lemma \ref{payment_identity}]
	For any $q'<q$, by BIC we have,
	
	\begin{eqnarray}
	\label{11}
	x^*(q)\hat{v}(q)-t^*(q) \geq x^*(q')\hat{v}(q)-t^*(q')\\
	\label{22}
	x^*(q')\hat{v}(q')-t^*(q') \geq x^*(q)\hat{v}(q')-t^*(q)
	\end{eqnarray}
	(\ref{11})+(\ref{22}) we get $$(x^*(q')-x^*(q))(\hat{v}(q')-\hat{v}(q)) \geq 0.$$ Since $\hat{v}(q') \geq \hat{v}(q)$, $x^*(q') \geq x^*(q)$, so $x^*(q)$ is monotone decreasing.
	
	For second part, we define $u(q)= x^*(q)\hat{v}(q)-t^*(q)$, (\ref{11}) becomes
	$$u(q')-u(q) \leq x^*(q')(\hat{v}(q')-\hat{v}(q)) $$
	
	(\ref{22}) becomes
	
	$$u(q')-u(q) \geq x^*(q)(\hat{v}(q')-\hat{v}(q)),$$ so
	
	$$x^*(q)(\hat{v}(q')-\hat{v}(q)) \leq u(q')-u(q) \leq x^*(q')(\hat{v}(q')-\hat{v}(q))$$
	
	substitute $q'$ by $q+\epsilon$, and let $\epsilon \rightarrow 0$, we get
	
	$$u'(q)=-x^*(q)\hat{v}'(q)$$
	
	Integrate over $q$ we have $u(q)= -\int_{r=q}^1 x^*(r)\hat{v}'(r)dr +u(1) $\  \
	
	So we have $ t^*(q)= \hat{v}(q)x^*(q)+\int_{r=q}^1 x^*(r)\hat{v}'(r) dr-\hat{v}(1)x^*(1)+t^*(1)$
	
\end{proof}

\begin{proof}[Proof of Lemma \ref{utilitity1}]
	
	By Definition \ref{induced_game} and Lemma \ref{payment_identity}
	
	\begin{eqnarray*}
		&&U=\int_0^1 (x^*(q)v(q)-t^*(q)) dq \\
		&&=\hat{v}(1)x^*(1)-t^*(1)+\int_0^1 (v(q)x^*(q)-x^*(q)\hat{v}(q))dq-\int_0^1 \int_{q}^1 x^*(r)\hat{v}'(r)drdq\\
		&&=\hat{v}(1)x^*(1)-t^*(1)+\int_0^1 x^*(q)(v(q)-\hat{v}(q))dq-\int_0^1 \int_0^r x^*(r)\hat{v}'(r) dqdr\\
		&&=\hat{v}(1)x^*(1)-t^*(1)+\int_0^1 x^*(q)(v(q)-\hat{v}(q))dq-\int_0^1 qx^*(q)\hat{v}'(q) dq\\
		&&=\hat{v}(1)x^*(1)-t^*(1)+\int_0^1 x^*(q)(v(q)-\hat{v}(q)-q\hat{v}'(q)) dq\\
		&&=\hat{v}(1)x^*(1)-t^*(1)+\int_0^1 x^*(q)(v(q)-r(q)) dq
	\end{eqnarray*}
\end{proof}

For discontinuous or unbounded cases, we have the following:

\begin{lemma}
	For any auctions considered in this payer, all results for any discontinuous or unbounded distribution are limitations of series of continuous and bounded distribution.
\end{lemma}
\begin{proof}
	For any fake distribution $\hat{v}(\cdot)$, consider the following distribution $f(\cdot)$:
	\item w.p. $1-\epsilon$, $f(q)=\hat{v}(q)$
	\item w.p. $\epsilon$, $f$ is a $[\hat{v}(1),\hat{v}(0)]$ uniform distribution.
	
	Then $f$ is a continuous function that as $\epsilon \rightarrow 0$, $f \rightarrow \hat{v}(q)$.
	
	For any unbounded distributions, it can be approached by taking its truncations at higher and higher values. The truncate distribution is bounded.
	
	All our results for unbounded or discontinues are approached arbitrarily well by the series of continuous and bounded distribution.
\end{proof}
\section{Missing proofs from section 4}
\begin{proof}[Proof of Theorem \ref{prior-independent}]
	Fix other buyers' reports $\hat{v}_{-i}(\cdot)$ and the quantile $q$ of buyer $i$.
	By definition of prior-independent mechanism family, the interim allocation $x^*(q)$ only depends on the value of $\hat{v}(\cdot)$ at quantile $q$, regardless the whole fake distribution $\hat{v}(\cdot)$. We write $x(\hat{v}(q))=x^*(q)$ to reflect this observation.
	
	By Lemma \ref{payment_identity}, the buyer's interim utility
	$$ u^*(q) = v(q)x(\hat{v}(q))-\hat{v}(q)x(\hat{v}(q))+\int_{0}^{\hat{v}(q)} x(z) dz$$
	
	Note that the term is independent to any other the value $\hat{v}(q')$ such that $q' \neq q$, so we can maximize $\int_0^1 u^*(r) dr$ case by case for each $r \in [0,1]$
	
	Define $x'(v)=\frac{dx}{dv}$. We have
	$$\frac{du^*(q)}{d\hat{v}}=x'(\hat{v}(q))(v(q)-\hat{v}(q))$$
	Since $x'(\hat{v}(q)) \geq 0$, $u_i^*(q)$ is maximized at $\hat{v}(q)=v(q)$. So the buyer will always want to report $\hat{v}(q)=v(q)$ at quantile $q$. In addition, the buyer's utility increases as $\hat{v}(q)$ approaches $v(q)$.
\end{proof}
\section{Missing proofs from section 5}
\begin{proof}[Proof of Theorem 5.2]
	For fixed  $\hat{v}_{-i}(\cdot)$, by Lemma 3.4, the buyer's best response $\hat{v}(q)$ maximizes
	\begin{equation}
	\label{F}
	\int_0^1 x(r(q))(v(q)-r(q))dq.
	\end{equation}

	Define $F(q)=x(r(q))(v(q)-r(q))$. The problem becomes to find a $r(\cdot)$ that maximizes $\int_0^1 F(q)dq$. We maximize $F(q)$ case by case for each $q \in [0,1]$. For any $q$, the solution is  $r(q)=\arg\max_{r(q)\geq 0} x(r(q))(v(q)-r(q)) , q \in [0,1]$.
	
	In particular,
	\begin{equation*}
	\hat{v}(q)=\frac{x(r(q))}{x'(r(q))}+r(q),
	\end{equation*}
	if $\frac{x(r(q))}{x'(r(q))}+r(q)\geq 0$, where $x'(r)=\frac{dx}{dr}.$
	
	
	Given such $r(q)$, we solve the differential equation $r(q)=\hat{v}(q)+q\hat{v}'(q)$,$q\in [0,1]$ to get $\hat{v}(q)$. One of the solutions is
	$$\hat{v}(q)=\frac{\int_0^q r(s)ds}{q}.$$
	
	In order to ensure that such $\hat{v}(q)$ is a distribution and is regular, we need to show the monotonicity of $\hat{v}(q)$ and $r(q)$. To prove the monotonicity of $r(q)$, for any $q'<q$, by definition of $r(q)$:
	$$x(r(q))(v(q)-r(q)) \geq x(r(q'))(v(q)-r(q'))$$
	$$x(r(q'))(v(q')-r(q')) \geq x(r(q))(v(q')-r(q))$$
	Add them together we have
	$$(x(r(q))-x(r(q')))(v(q)-v(q')) \geq 0$$
	Since $v(q) \leq v(q')$, by the monotonicity of $x(r)$, we get $r(q) \leq r(q')$.

	Given the monotonicity of $r(q)$, note that
	$$\hat{v}(q)=\frac{\int_0^q r(s)ds}{q} \geq \frac{\int_0^q r(q)ds}{q} =\frac{qr(q)}{q}=r(q)$$
	And by $r(q)=\hat{v}(q)+q\hat{v}'(q)$ we get $\hat{v}'(q)<0$.

\end{proof}
\begin{proof}[Proof of Example \ref{Myerson_example1}]

For buyer one's quantile $q$, $F_1(q)$ in (\ref{F}) is $\frac{r_1(q)+1}{2}(v_1(q)-r_1(q))$,

The optimal non-negative $r_1(q)$ to maximize $F_1(q)$ is $r_1(q)=0$, for all $q\in [0,1]$.

So the buyer 1's best response in the is $\hat{v}_1(q)=\epsilon$ (meanwhile, $r(q)=\epsilon$ ). Note that $r_2(q)<0$ when $q \in (\frac{1}{2},1]$. Let $\epsilon \rightarrow 0$, buyer 1 guarantees a constant allocation probability $x_1^*(q)=\frac{1}{2}$ for any $q$, with the payment tends to 0.

\end{proof}
\begin{proof}[Proof of Example \ref{example_2_buyer}]
	Given the other buyer's fake distribution $\hat{v}_{-i}(q)=\frac{1}{2}-\frac{1}{4}q$, so the corresponding virtual value curve $r_{-i}(q)=\frac{1-q}{2}$. So $x(r(q))=Pr_{q\in [0,1]}[\frac{1-q}{2} \leq r(q)]=2r(q)$. By (\ref{F}), buyer $i$'s utility is
	$$\int_0^1 2r(q)(v(q)-r(q))$$
	
	Fix $q$, $2r(q)(v(q)-r(q))$ get its maximum if $r(q)=\frac{v(q)}{2}=\frac{1-q}{2}$
	
	Solving the differential equation $\hat{v}(q)+q\hat{v}'(q)=\frac{1-q}{2}$ we get the solution $\hat{v}(q)=\frac{1}{2}-\frac{1}{4}q$
	
\end{proof}

In the rest of this section, we give the proof of Theorem \ref{equal_to_FPA}

Note that the utility of Myerson mechanism family only depends on the virtual value function $r(\cdot)$, i.e., for a set of fake distributions that corresponding to a same virtual value function, (the differential function $\hat{v}(q)+q\hat{v}'(q)=r(q)$ has multiple solutions of $\hat{v}(q)$), the buyer has the same utility. So we can simplify the buyer's action to reporting the virtual value function $r(q)$, instead of reporting the fake distribution $\hat{v}(q)$. 

To describe the main theorem of this section, we first rewrite the induced game of Myerson mechanism family in an equivalent way as follows:
\begin{definition}
	
	The induced game of Myerson mechanism family given true distribution profile $\mathbf{v}(\cdot)$ is a normal-form game $(N,A,U)$ where
	
	\begin{itemize}
		\item $N=\{1,...,n\}$ buyers, indexed by $i$. 
		\item $A=A_1 \times A_2...\times A_n$, where  $A_1=A_2=...=A_n=\mathcal{F}$. We use $r_i(\cdot) \in \mathcal{F}$ to denote an action of buyer $i$, in other words, the action of buyer $i$ is to choose a virtual value distribution.
		\item $U=\{U_1,...,U_n\}$,  where $U_i : A \rightarrow \mathbb{R}$ denote the utility function of buyer $i$ such that (from \ref{F})
		$$U_i(r_1(\cdot),r_2(\cdot),...,r_n(\cdot))=\int_0^1  (\int_{q_{-i}} x^{M}_i(\mathbf{r}(\mathbf{q})) dq_{-i})(v_i(q_i)-r_i(q_i)) dq_i$$
		where $\mathbf{x}^M$ is the allocation vector in Myerson mechanism, taking $n$ buyers' virtual values as inputs.
	\end{itemize}
\end{definition}
Consider a standard first-price auction (FPA) with type distribution profile $\mathbf{v}(\cdot)$, which is a Bayesian game. The buyer's bidding strategy is a function $b(q): [0,1] \rightarrow \mathbb{R}$ that maps each quantile (equivalently, value) to bid.

\begin{lemma}
	In first-price auction, given buyers' bidding strategy profile $\mathbf{b}(\cdot)$, buyer $i$'s expected utility is
	$$\int _0^1 (\int_{q_{-i}} x^{F}_i(\mathbf{b}(\mathbf{q})) dq_{-i})(v_i(q_i)-b_i(q_i)) dq_i$$ 
	where $\mathbf{x}^F$ is the allocation vector in first-price auction, given $n$ bids as inputs.
\end{lemma}

\begin{proof}
	In FPA, for a fixed quantile profile $\mathbf{q}$, the buyer pays his/her bid if he/she get the item, and pay 0 otherwise, so the utility is $x_i^F(\mathbf{b}(\mathbf{q}))(v_i(q_i)-b_i(q_i))$. As an result, the expected utility is 
	$$\int_{\mathbf{q}} x_i^F(\mathbf{b}(\mathbf{q}))(v_i(q_i)-b_i(q_i))d\mathbf{q}=\int _0^1 (\int_{q_{-i}} x^{F}_i(\mathbf{b}(\mathbf{q})) dq_{-i})(v_i(q_i)-b_i(q_i)) dq_i$$ 
\end{proof}

In th same spirit as we defined for induced game of Myerson mechanism family, we can regard the first-price auction as a normal-form game, where each player's action is simply choosing a bidding strategy, and the utility is the expected utility defined above. Formally,
\begin{definition} 
	The induced game of first-price auction given type distribution profile $\mathbf{v}(\cdot)$ is normal-form game $(N,A,U)$ where
	
	\begin{itemize}
		\item $N=\{1,...,n\}$ buyers, indexed by $i$
		\item $A=A_1 \times A_2...\times A_n$, where  $A_1=A_2=...=A_n=\mathcal{F}$. We use $b_i(\cdot) \in \mathcal{F}$ to denote an action of buyer $i$, in other words, the action of buyer $i$ is to choose a bidding strategy.
		\item $U=\{U_1,...,U_n\}$,  where $U_i : A \rightarrow \mathbb{R}$ denotes the utility function of buyer $i$ such that 
		$$U_i(b_1(\cdot),b_2(\cdot),...,b_n(\cdot))=\int_0^1  (\int_{q_{-i}} x^{F}_i(\mathbf{b}(\mathbf{q})) dq_{-i})(v_i(q_i)-b_i(q_i)) dq_i$$
	\end{itemize}
\end{definition}

It is natural to only consider decreasing bidding strategy in FPA, since it is not hard to show that any non-decreasing bidding strategy is weakly dominated by a decreasing bidding strategy.

\begin{theorem}
	\label{u1=u2}
	Given $\mathbf{v}(\cdot)$, the induced game of Myerson mechanism family $(N,\mathcal{F}^n,U^1)$ is identical to the induced game of FPA $(N,\mathcal{F}^n,U^2)$.
\end{theorem}
\begin{proof}[Proof of Theorem \ref{u1=u2}]

	To verify  that $U^1$ and $U^2$ are the same function, we need to prove that given  the same input, i.e, $\mathbf{b}(\cdot)=\mathbf{r}(\cdot)$, we have $U^1(\mathbf{b}(\cdot))=U^2(\mathbf{r}(\cdot))$.
	
	In fact, $x^*_F$ and $x^*_M$ are exactly the same function: Given a quantile profile $\mathbf{q}$, FPA allocates the item to the buyer with the largest bid $b(q)$, and the Myerson mechanism allocates the item to the buyer with the largest virtual value $r(q)$. As $\mathbf{b}(\cdot)=\mathbf{r}(\cdot)$, the two utilities are indeed the same. This proves the theorem.

\end{proof}

\section{Missing proofs from section 6}

Fix $\hat{v}_{-i}(\cdot)$ and the reserve price of $\hat{v}(\cdot)$, given the quantile $q$ for the buyer, following a similar argument as before,  the interim allocation only depends on the value of $\hat{v}(q)$. We write $x^*(q)=x(\hat{v}(q))$.

\subsection{Missing proofs from section 6.1}

First let $q^*$ denote the reserve quantile of $v(\cdot)$ (where the revenue curve $qv(q)$ get its maximum), and define $R^*=q^*v(q^*)$. $q^*$ and $R^*$ are constants given the true distribution. Let $\hat{q}^*$ denote the reserve quantile of $\hat{v}(\cdot)$, and define $R=\hat{q}^*\hat{v}(\hat{q}^*)$. $\hat{q}^*$ and $R$ are to be determined. 

To prove Theorem \ref{vqformall} and \ref{no_equilibrium}, we introduce following lemmas:
\begin{lemma}
	\label{vq}
	Subject to $\hat{q}^*,R$ unchange, the buyer's utility weakly increases as $\hat{v}(\cdot)$ approaches $v(\cdot)$.
\end{lemma}

Note that by Lemma \ref{ironing}, the reserve price $\hat{v}(q^*)=\frac{R}{\hat{q}^*}$ is also fixed. Directly from the proof of Theorem 4.2, $u^*(q)$ increases for all $q<q^*$, as $\hat{v}(q)$ approaches $v(q)$. So the buyer's utility also increases.



\begin{lemma}
	\label{R_less_than_Rstar}
	It is without loss of generality to consider $\hat{v}(\cdot)$ subject to $R \leq R^*$.
\end{lemma}
\begin{proof}
	Suppose $R>R^*=v(q^*)q^*$. We change the distribution $\hat{v}(\cdot)$ as follows: for all $q<\hat{q}^*$ such that $q\hat{v}(q)>R^*$, set $\hat{v}(q)=\frac{R^*}{q}$. After the change, every $\hat{v}(q)$ get closer to $v(q)$ and the reserve price decreases. (Note that after the change the reserve quantile $\hat{q}^*$ remains the same.) By Theorem 4.2, each interim utility $u^*(q)$ increases for all $q<q^*$. Since for $q>q^*$, the utility is always 0, we can conclude that the overall utility increases.
\end{proof}

So we only need to consider the case where $R \leq R^*$. Since $v(\cdot)$ is a regular distribution, $qv(q)$ is concave with $q$. Let $q_1,q_2$ denote the two intersection points between revenue curve $qv(q)$ and straight line $y=R$. (If $R=R^*$, set $q_1=q_2$. If $v(1)>R$, set $q_2=1$). It's easy to see that it is not optimal when $\hat{q}^*<q_2$. Let $q_3$ denote the quantile of the intersection between revenue curve $y=qv(q)$ and straight line $((0,0),(\hat{q}^*,R))$. If $v(1)>R$, set $q_3=1$, otherwise $q_3$ always exists and $q_2<q_3<1$.

\begin{lemma}
	\label{vqform}
	For fixed $R$ subject to $R \leq R^*$ and fixed reserve quantile $\hat{q}^*$, it is without loss of generality to consider $\hat{v}(\cdot)$ with the following form.
	
	\begin{itemize}
		\item For $q$ s.t. $\hat{q}^*<q<1$, $\hat{v}(q)$ can be any value less than $\frac{R}{q^*}$,
		\item For $q$ s.t. $q_3<q<q^*$, $\hat{v}(q)=\frac{R}{\hat{q}^*}$ (constant),
		\item For $q$ s.t. $q_2<q<q_3$ or $0<q<q_1$, $\hat{v}(q)=v(q)$ (truthful)
		\item For $q$ s.t. $q_1<q<q_2$, $\hat{v}(q)=\frac{R}{v}(q)$ (equal revenue curve)
	\end{itemize}
	(the revenue curve of $\hat{v}(q)$ is the red line (line $Oq_1q_2q_3\hat{q}^*$) in Figure 1.
\end{lemma}

\begin{proof}
	Given $R$ and $\hat{q}^*$ fixed, the feasible domain of $\hat{v}(q)$ is $[\frac{R}{\hat{q}^*},\frac{R}{q}]$, i.e. the revenue curve $q\hat{v}(q)$ in figure 1 should below the line $y=R$ and above the line $((0,0),(\hat{q}^*,R))$. Then by Lemma \ref{vq} we should choose the value in the feasible region that is the closest to $v(q)$, so the revenue curve of the best response distribution is the red line (line $Oq_1q_2q_3\hat{q}^*$).
\end{proof}

\begin{figure}[htbp]
	\centering
	\includegraphics[width=0.5\textwidth]{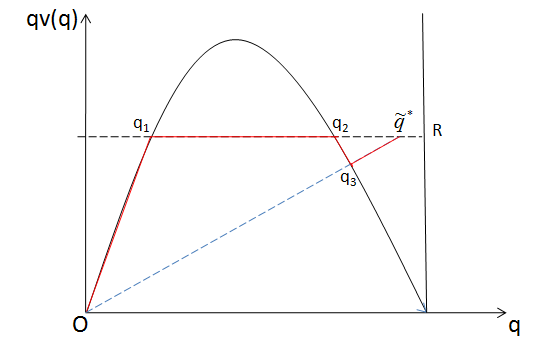}
	\caption{Best response to SPA with Myerson reserve}
	\label{fig:digit}
\end{figure}

\begin{lemma}
	
	\label{qform}
	For fixed $R$, and suppose $\hat{v}(q)$ has the form in Lemma \ref{vqform}, it is without loss of generality to set $\hat{q}^*=1$. 
\end{lemma}

\begin{proof}
		By Lemma \ref{vqform} we assume w.l.o.g. that $R<R^*$. $q_1,q_2,q_3$ are define in Figure 1. We have $\frac{q_3}{\hat{q}^*}=\frac{q_3v(q_3)}{R}$. Then $\hat{q}^*=\frac{R}{v(q_3)}$, which is uniquely determined by $q_3$.
		
		
		\begin{eqnarray*}
			U &=&\int_0^{q_1}  x(v(q))(-qv'(q))dq+\int_{q_1}^{q_2} x(\frac{R}{q})v(q)dq \\
			&+&\int_{q_2}^{q_3} x(v(q))(-qv'(q))dq+\int_{q_3}^{\frac{R}{v(q_3)}}x(v(q_3)(v(q)-v(q_3)))dq
		\end{eqnarray*}
		We compute the derivative of $U$ to $q_3$
		\begin{eqnarray*}
			U' &=& -x(v(q_3))q_3v'(q_3)+\int_{q_3}^{\frac{R}{v(q_3)}}(x'(v(q_3))v'(q_3)(v(q)-v(q_3))\\
			&&-v'(q_3)x(v(q_3)))dq-\frac{Rv'(q_3)}{v^2(q_3)}x(v(q_3))(v(\frac{R}{v(q_3)})-v(q_3))\\
			\frac{U'}{v'(q_3)}&=&-q_3x(v(q_3))+\int_{q_3}^{\frac{R}{v(q_3)}}(x'(v(q_3))(v(q)-v(q_3))-x(v(q_3)))dq-\\
			&&\frac{R}{v^2(q_3)}x(v(q_3))(v(\frac{R}{v(q_3)})-v(q_3))\\
			&\leq & -q_3x(v(q_3))-x(v(q_3))(\frac{R}{v(q_3)}-q_3)-\frac{R}{v^2(q_3)}x(v(q_3))(v(\frac{R}{v(q_3)})-v(q_3))\\
			\frac{U'}{v'(q_3)x(v(q_3))} &\leq& -\frac{R}{v(q_3)}-\frac{R}{v^2(q_3)}(v(\frac{R}{v(q_3)})-v(q_3))\\
			\frac{U'v(q_3)}{v'(q_3)x(v(q_3))R} &\leq& -1-\frac{v(\frac{R}{v(q_3)})-v(q_3)}{v(q_3)}=-\frac{v(\frac{R}{v(q_3)})}{v(q_3)}<0\\
		\end{eqnarray*}
		(The first inequality holds because $x'(v)>0$ for any $v$ and $v(q)<v(q_3)$ for $q\in [q_3,\frac{R}{v(q_3)}]$)
		
		So we can conclude that $U'>0$ since $v'(q_3)<0$ and other terms in left side are non-negative. $\hat{q}^*=1$ when $q_3$ gets its feasible maximum.
	
\end{proof}

So Theorem \ref{vqformall} is proved.

\begin{lemma}
	\label{vqutility}
	For fixed $R$, the buyer's utility by reporting the best response distribution described in Lemmas \ref{vqform} and \ref{qform} is given by,
	$$
	U =\int_0^{q_1}  x(v(q))(-qv'(q))dq+\int_{q_1}^{q_2} x(\frac{R}{q})v(q)dq+\int_{q_2}^{q_3} x(v(q))(-qv'(q))+\int_{q_3}^1x(R)(v(q)-R)dq$$
	where $q_1v(q_1)=R,q_2(q_2)=R,v(q_3)=R$.
	
\end{lemma}


\begin{lemma}
	\label{minus_e}
	For any buyers' true distributions that are i.i.d. regular , if $v(1)=0$ and the best response given by Lemma \ref{vqform} and \ref{qform} is unique, then there does not exist a symmetric equilibrium of fake distribution.
\end{lemma}

\begin{proof}
	$v(1)=0$ implies that $q_2,q_3<1$. Assume by contradiction that if there exists a symmetric equilibrium, we use $\hat{v}(q)$ to denote the fake distribution in equilibrium. By Lemma \ref{vqform} and \ref{qform}, the equilibrium of fake distribution have a point mass on $[q_3,1]$, where all buyers' fake types are equal to the reserve price $\hat{v}(1)=R$. The mechanism must deal with tie-breaking in this case. Note that when $q\in [q_3,1]$, this part of utility is $\int_{q_3}^1 x(R)(v(q)-R) dq<0$. So, to maximize utility, the buyer wants to minimize $x(R)$. Whatever the tie breaking rule is, all buyers can always deviate by using $R-\epsilon$ instead of $R$. The first three terms in Lemma \ref{vqutility} will vary by an infinitesimal, while in the last term, $x(R)$ vary from a negative non-infinitesimal to $0$, and the whole term will increase by a non-infinitesimal amount. So it is not an equilibrium, unless all $\hat{v}(q)$ equal 0. But it's easy to check that $\hat{v}(q)=0$ is not an equilibrium either. Contradiction.
\end{proof}

It remains to show that the best response of fake distribution is unique. We have following lemmas.

\begin{lemma}
	\label{no1}
	For any buyers' true distributions that are i.i.d. regular and do not have a point mass, the symmetric equilibrium of fake distribution can not have a point mass.
\end{lemma}
\begin{lemma}
	\label{no2}
	For any buyers' true distributions that are i.i.d. regular and do not have a point mass, if there exists a symmetric equilibrium of best response of fake distribution with fixed $R$ and reserve quantile $\hat{q}^*$, then the best response of fake distribution given by Lemma \ref{vqform} is unique.
\end{lemma}

Theorem \ref{no_equilibrium} then follows from Lemmas \ref{no1} and \ref{no2}.

\begin{proof}[Proof of Lemma  \ref{no1} and \ref{no2}]
	We first assume by contradiction that there is a point mass on equilibrium of fake distribution, i.e. $\hat{v}(q)=c$ for $q \in [a,b],a<b$.
	
	Since $v(q)$ does not have a point mass, there exist a subinterval $[x,y]$ s.t. $v(q)>c,~\forall q \in [x,y]$ or $v(q)<c ,~\forall q \in [c,d]$.
	
	We consider case by case. If $v(q)>c ,\forall q \in [x,y]$, one can deviate the fake distribution by $\hat{v}(q)=c+\epsilon-\epsilon\frac{q-x}{y-x}, \forall q \in [x,y]$ (slightly adjust the value in the left-neighborhood of $x$ to keep the monotonicity). This slight disturb remains $\max q\hat{v}(q)$ and $\arg\max q\hat{v}(q)$ the same, and note that since there is a tie breaking rule for $q\in [x,y]$, it gives a non-infinitesimal increasing to the utility.
	
	If $v(q)<c ,~\forall q \in [x,y]$, for cases $\arg\max q\hat{v}(q) \neq y$, one can deviate the fake distribution by $\hat{v}(q)=c-\epsilon$ (slightly adjust the value in the neighborhood of  endpoints $x,y$ to keep the monotonicity). And for case $\arg\max q\hat{v}(q) = y$, one can deviate by setting $R=R-\epsilon$ and according to Lemma \ref{minus_e}. All these slight disturb gives a non-infinitesimal increasing to the utility.
	
	To show Lemma \ref{no2}, fix $\max q\hat{v}(q)$ and $\arg\max q\hat{v}(q)$. Assume by contradiction that in symmetric equilibrium there exists a fake distribution $\hat{v}(q)$ that does not consistent to the distribution we give in Lemma \ref{vqform} (denote by $\hat{v}^*(q)$) but give the same utility. Suppose that $\hat{v}^*(q) \neq \hat{v}(q)$ for a given quantile $q$, by Lemma \ref{vq} and Theorem 4.2, we can infer that $x(\hat{v}(q))=x(\hat{v}^*(q))=Constant,\forall v\in [\hat{v}^*(q),\hat{v}(q)]$. However, since there is no point mass, in symmetric equilibrium $x(\hat{v}(q))=q$, i.e., no constant allocation in any interval, contradiction.
	
	Combining Lemma \ref{vqform}, Lemma \ref{no1} and Lemma \ref{no2}, we conclude that in symmetric environment there is no equilibrium of fake distribution. Nevertheless, Lemma \ref{qform} can be used to solve equilibriums in asymmetric cases.
\end{proof}

\subsection{Missing proofs from section 6.2}

In this setting, the revenue curve of fake distribution is a concave function. Let $\hat{q}^*$ denote the reserve quantile (where the $q\hat{v}(q)$ is maximized), then from $0$ to $\hat{q}^*$, $q\hat{v}(q)$ is weakly increasing. Define $R=\hat{q}^*\hat{v}(\hat{q}^*)$.

\begin{lemma}
	\label{rqform}
	It is without loss of generality to consider $\hat{v}(\cdot)$  subject to $\hat{q}^*=1$. 
\end{lemma}

\begin{proof}
	If $\hat{q}^*<1$, set $\hat{v}(q)=\frac{R}{q}$ for all $q \in [\hat{q}^*,1]$. The utility increases by $\int_{\hat{q}^*}^1 x(\frac{R}{q})v(q) dq >0$. So $\hat{q}^*=1$
\end{proof}

\begin{proof}[Proof of Theorem \ref{rvqform}]
	First of all, by a similar argument as in Lemma \ref{R_less_than_Rstar}, we have $R \leq R^*$.
	
	
	By Lemma \ref{rqform}, we can assume WLOG that  $\hat{q}^*=1$. So the revenue curve $q\hat{v}(q)$ is monotone increasing in [0,1]. Let $q_2$ denote the quantile of the intersection between $q\hat{v}(q)$ and the right part of revenue curve $qv(q)$ ($q \in [q^*,1]$). If $q_2$ does not exist (i.e. $v(1)>\hat{v}(1)$), set $q_2=1$.
	
	For all $q\in [0,q_2]$, since $q\hat{v}(q)$ is increasing, $\hat{v}(q)\leq \frac{q_2\hat{v}(q_2)}{q}$. By Theorem 4.2, $\hat{v}(q)$ should be the value in this region that is the closest to $v(q)$, i.e. $\hat{q}=\min \{\frac{q_2\hat{v}(q_2)}{q},v(q)\}$. Let $q_1$ denote the quantile of the intersection between straight line $y=q_2 \hat{v}(q_2)$ and the left part of revenue curve $qv(q),q \in [0,q^*]$. So $\hat{v}(q)=v(q)$ for $q \in [0,q_1]$ and  $\hat{v}(q)=\frac{q_2 \hat{v}(q_2)}{q}$ for $q \in [q_1,q_2]$.
	
	For $q \in  [q_2,1]$, set $\hat{v}(q)=\frac{q_2 \hat{v}(q_2)}{q}$, which is the minimum possible value of revenue curve. After this change, the reserve price $\hat{v}(1)$ weakly decreases, and all $\hat{v}(q)$ gets closer to $v(q)$. By Theorem 4.2 the total utility increases.
	
	To sum up, the fake distribution $\hat{v}(\cdot)$ has the form in the Theorem 6.14, and satisfies concavity and monotonicity.
\end{proof}

When the buyer fixes the incident quantile, the best response of fake distribution is determined.

For each buyer $i$, let $q_i$ denote the incident quantile.
\begin{lemma}
	Buyer $i$'s utility of reporting the best response in Lemma Theorem \ref{rvqform} is
	$$U_i=\int_0^{q_i} x(v(q))(-qv'(q))dq+\int_{q_i}^1 x(v(\frac{q_iv(q_i)}{q}))v(q) dq$$
	
	In an equilibrium of the induced game, $\frac{dU_1}{q_1}=\frac{dU_2}{q_2}=...=\frac{dU_n}{q_n}=0$
\end{lemma}

\begin{lemma}
	\label{incident_quantile}
	Suppose $n=2$ and true distributions are i.i.d. regular, for symmetric equilibrium of the induced game, the incident quantile $q_0$ (identical to both buyers) satisfies
	$$q_0-1+\frac{1}{q_0v(q_0)}\int_{q_0}^1 qv(q) dq=0$$
\end{lemma}

\begin{proof}
	Consider the derivative of $U$ to $q_0$.
	\begin{eqnarray}
	\label{dudq0}
	\frac{dU}{dq_0}=x(v(q_0))(-q_0v'(q_0))+\int_{q_0}^1 \frac{v(q)dx(v(\frac{q_0v(q_0)}{q}))}{dq_0}dq-x(v(q_0))v(q_0)
	\end{eqnarray}
	Note that in a symmetric equilibrium, $x(v)=1-v_{-i}^{-1}(v)=1-q$, and $x'(v)  \triangleq \frac{dx}{dv}=-\frac{1}{v'(q)}$. Define $r(q)=v(q)+qv'(q)$. For $q \in [q_0,1]$
	$$\frac{dx(v(\frac{q_0v(q_0)}{q}))}{dq_0}=\frac{dx}{dv}\frac{dv}{d(\frac{q_0v(q_0)}{q})}\frac{d(\frac{q_0v(q_0)}{q})}{dq_0}
	=-\frac{1}{v'(q)}\frac{dv}{dq}\frac{dq}{d(\frac{q_0v(q_0)}{q})}\frac{r(q_0)}{q}=\frac{q^2}{q_0v(q_0)}\frac{r(q_0)}{q}=\frac{qr(q_0)}{q_0v(q_0)}$$
	So
	$$\frac{dU}{r(q)dq_0}=-x(v(q_0))+\int_{q_0}^1 \frac{qv(q)}{q_0v(q_0)} dq=q_0-1+\frac{1}{q_0v(q_0)}\int_{q_0}^1 qv(q) dq$$
	
	Let the formula above be 0 concludes the proof.
\end{proof}

\begin{proof}[Proof of Example \ref{example_SPAMR}]
	By Lemma \ref{incident_quantile}, we have $q_0-1+\frac{1}{q_0(1-q_0)}(\frac{1}{6}-\frac{1}{2}q_0^2+\frac{1}{3}q_0^3)=0$
	
	Solving this equation we get $q_0=\frac{1}{4}$.
	
\end{proof}

The revenue in the example is $REV=2 \int_0^{\frac{1}{4}} (1-q)(1-2q) dq=\frac{1}{3}$

Compared to SPA in the standard setting, buyer's utility and social welfare increase, while the revenue remains the same. This is not a coincidence. In the following, we prove Theorem \ref{revenue_equal} that says for i.i.d., regular true distributions, the revenue of the equilibrium in SPAMR equals to the revenue of SPA in standard setting.


\begin{proof}[Proof of Theorem \ref{revenue_equal}]
	For simplicity here we only prove the case when $n=2$.
	
	By Theorem \ref{rvqform} and Lemma \ref{revenue}
	\begin{eqnarray*}
		REV&=&2\int_0^{q_0} (1-q)(v(q)+qv'(q))dq = 2\int_0^{q_0} (1-q) d(qv(q))\\
		&=&2q_0(1-q_0)v(q_0)+2\int_0^{q_0} qv(q) dq
	\end{eqnarray*}
	By Lemma 6.18, we have $(1-q_0)q_0v(q_0)=\int_{q_0}^1 qv(q) dq$
	Combine together we get $REV=2 \int_0^1 qv(q) dq$
	Note that the revenue of second-price auction is $$REV_2=2\int_0^1 (v(q)+qv'(q)) dq=2\int_0^1 qv(q) dq,$$
	
	So $REV=REV_2$
\end{proof}
This lemma holds for the $n$ i.i.d. buyers case, which says that:

\begin{lemma}
	For $n$ buyers with i.i.d. regular true type distribution, the revenue of second-price with Myerson's reserve mechanism family on equilibrium of fake distribution (regular fake distribution only) equal to that of the second-price auction.
\end{lemma}
\begin{proof}
	Using the same notations as Lemma \ref{incident_quantile}. In this case,
	
	$x^*(q)=(1-v^{-1}(\frac{q_0v{q_0}}{q}))^{n-1}$, so
	\begin{eqnarray*}
		\frac{dx}{dq_0}&=&(n-1)(1-q)^{n-2}\frac{dq}{d(\frac{q_0v(q_0)}{q})}\frac{d(\frac{q_0v(q_0)}{q})}{dq_0}\\
		&=&(n-1)(1-q)^{n-2}\frac{q^2}{q_0v(q_0)}\frac{r(q_0)}{q}=(n-1)(1-q)^{n-2} \frac{qr(q_0)}{q_0v(q_0)}
	\end{eqnarray*}
	Substituting this into (\ref{dudq0}) and let $\frac{dU}{dq_0}=0$ we get
	$$(1-q_0)^{n-1}=\int_{q_0}^1 \frac{(n-1)(1-q)^{n-2}qv(q)}{q_0v(q_0)}dq$$
	The revenue
	\begin{eqnarray*}
		REV&=&n\int_0^{q_0} (1-q)^{n-1}(v(q)+qv'(q))dq\\
		&=&nq_0v(q_0)(1-q_0)^{n-1} -n\int_0^{q_0} (n-1)(1-q)^{n-2}qv(q) dq\\
		&=& n\int_0^1 (n-1)(1-q)^{n-2}qv(q) dq =REV_2
	\end{eqnarray*}
\end{proof}

\section{Missing proofs form section 7}

Fix $\hat{v}_{-i}(\cdot)$, for any realized value of $q_r$, as long as $q \geq q_r$, for the same reason as before, we can write $x^*(q)=x(\hat{v}(q))$, or $x^*(q)=0$ if $q < q_r$.

\begin{proof}[Proof of Theorem \ref{ssvqform}]
	When the reserve quantile is $q_r$, by definition, the buyer utility is $\int_0^{q_r} x(\hat{v}(q))(v(q)-\hat{v}(q)-q\hat{v}'(q))dq$, so the expected utility is
	\begin{eqnarray*}
		U&=&\int_0^1 \int_0^{q_r} x(\hat{v}(q))(v(q)-\hat{v}(q)-q\hat{v}'(q)) dq dq_r\\
		&=&\int_0^1 dq \int_q^1 x(\hat{v}(q))(v(q)-\hat{v}(q)-q\hat{v}'(q))dq_r \\
		&=&\int_0^1 (1-q)x(\hat{v}(q))(v(q)-\hat{v}(q)-q\hat{v}'(q))dq
	\end{eqnarray*}

	This is a functional optimization problem. Define $F=(1-q)x(\hat{v}(q))(v(q)-\hat{v}(q)-q\hat{v}'(q))$. The Euler-Lagrange equation yields $$\frac{\partial F}{\partial \hat{v}}-\frac{d}{dq}\frac{\partial F}{\partial \hat{v}'}=0$$
	
	(Sometimes we omit the  independent variables   in differential functions)
	
	Define $x'(\hat{v})=\frac{dx}{d\hat{v}}=-\frac{1}{v_{-i}'(\hat{v}_{-i}^{-1}(\hat{v}))}$, for $n=2$ and in symmetric equilibrium
	$x'(\hat{v})=-\frac{1}{\hat{v}'(q)}$. We get
	$$x'(\hat{v})(1-q)(v-\hat{v}-q\hat{v}')-(1-q)x(\hat{v})+x'(\hat{v})\hat{v}'q(1-q)+x(\hat{v})(1-2q)=0$$
	By simplification we have
	$$v(q)=\hat{v}(q)-q\hat{v}'(q)$$
	Set a particular solution $\hat{v}(1)=v(1)$ we can get the solution of the differential equation.
	
	To show the monotonicity of $\hat{v}(\cdot)$, taking the derivative of $v(q)=\hat{v}(q)-q\hat{v}'(q)$ we have $v'(q)=-\hat{v}''(q)$, so $\hat{v}''(q) \geq 0$. \\
	Set $q=1$ we have $v(1)=\hat{v}(1)-\hat{v}'(1)$, so $\hat{v}'(1)=0$, combining with $\hat{v}''(q)  \geq 0$ we get $v'(q) \leq 0$.
	So $\hat{v}(q)$ is a well defined distribution.
	
	Actually, this is functional maximization problem in which both endpoints $\hat{v}(0),\hat{v}(1)$ are variables. So the function $\hat{v}(q)$ need to satisfy the natural boundary condition:
	$$\frac{\partial F}{\partial \hat{v}'}|_{q=0,1}=0$$
	It is automatically satisfied.
	
	As a result, when all buyers' fake distributions satisfy $v(q)=\hat{v}(q)-q\hat{v}'(q)$, it is an equilibrium. Since $\hat{v}(\cdot)$ is feasible, and satisfy the Euler-equation and natural boundary condition.
\end{proof}

\begin{proof}[Proof of Theorem \ref{singel_sample_less_than_SPA}]
	We first prove the case of $n=2$.
	
	In SPARQR, the revenue from one buyer is,
	
	\begin{eqnarray*}
		REV &=& \int_0^1 (1-q)x(\hat{v})(\hat{v}+q\hat{v}'(q)) dq=\int_0^1 (1-q)^2 d (q\hat{v}(q)) \\
		&=& \int_0^1 q\hat{v}(q)d(1-q)^2= \int_0^1 2q(1-q)\hat{v}(q) dq
	\end{eqnarray*}
	The revenue from one buyer in second-price auction is
	\begin{eqnarray*}
		REV_2 &=& \int_0^1 (1-q)(v(q)+qv'(q)) dq=\int_0^1 (1-q)d (qv(q)) \\
		&=& \int_0^1 qv(q) dq =\int_0^1 (q\hat{v}(q)-q^2\hat{v}'(q))dq\\
		&=& \int_0^1 q\hat{v}(q)dq-\int q^2 d\hat{v} =\int_0^1 q\hat{v}(q)dq-q^2\hat{v}(q)|_0^1+\int_0^1 2q\hat{v}(q)dq\\
		&=& \int_0^1 3q\hat{v}(q)dq-\hat{v}(1)
	\end{eqnarray*}
	So
	\begin{eqnarray*}
		REV_2-REV&=&  \int_0^1 (q+2q^2)\hat{v}(q) dq-\hat{v}(1)= \int_0^1 \hat{v}(q) d(\frac{1}{2}q^2+\frac{2}{3} q^3)-\hat{v}(1) \\
		&=& \frac{1}{6}\hat{v}(1)-\int_0^1 (\frac{1}{2}q^2+\frac{2}{3} q^3)\hat{v}'(q) dq \geq 0
	\end{eqnarray*}
\end{proof}
For cases where $n>2$:
\begin{lemma}
	For $n$ buyers with i.i.d. true type distribution, the revenue of single sample mechanism family on equilibrium of fake distribution is no more than the revenue of the SPA under the true distribution.
\end{lemma}
\begin{proof}
	Using the same notations as Theorem 7.4. , In this case,
	
	$x^*(v)=(1-v(q)^{-1}(v))^{n-1}=(1-q)^n$, $x'(v)=-\frac{1}{v'(q)}(n-1)(1-q)^{n-2}$
	
	The solution of Euler equation gives
	
	\begin{eqnarray}
	\label{euler2}
	v=-\frac{q\hat{v}'}{n-1}+\hat{v}
	\end{eqnarray}
	
	Set the particular solution $\hat{v}(1)=v(1)$, then we get the solution $\hat{v}(\cdot)$.
	
	Then we show the monotonicity of $\hat{v}(\cdot)$. Assume by contradiction that there exists a $q_1$ s.t. $v'(q_1)<0$. From (\ref{euler2}) we have $\hat{v}(q_1)>v(q_1)$.
	
	Consider the first minimum quantile $q_2 \in (q_1,1)$ s.t. $\hat{v}(q_2) \leq v(q_2)$, since $\hat{v}(1)=v(1)$ $q_2$ must exists. Then for all $q \in (q_1,q_2)$, by (\ref{euler2}) $\hat{v}'(q)>0$. So $\hat{v}(q_2)>\hat{v}(q_1)$. Therefore $v(q_2) \geq \hat{v}(q_2) >\hat{v}(q_1)>v(q_1)$, it contradicts with the decreasing monotonicity of $v(\cdot)$.
	
	So we proved that $\hat{v}(\cdot)$ is the feasible solution (Note that the nature boundary condition is automatically satisfied).
	
	\begin{eqnarray*}
		REV&= &\int_0^1 (1-q)x(\hat{v})(\hat{v}(q)+q\hat{v}'(q))=\int_0^1 (1-q)^n d(q\hat{v}(q))=q\hat{v}(q)n(1-q)^{n-1} dq
	\end{eqnarray*}
	
	\begin{eqnarray*}
		REV_2 &=& \int_0^1 qx(v(q))(v(q)+qv'(q)) dq=\int_0^1 q(1-q)^{n-1} d(qv(q))\\
		&=& \int_0^1 (n-1)q(1-q)^{n-2}v(q) dq \\
		&= &\int_0^1 (-q(1-q)^{n-2}\hat{v}'(q)+q(1-q)^{n-2}(n-1)\hat{v}(q)) dq \\
		&=& \int_0^1 \hat{v}(q)(-q^2(n-2)(1-q)^{n-3}+2q(1-q)^{n-2}+q(1-q)^{n-2}(n-1)) dq
	\end{eqnarray*}
	Therefore, when $n>=3$,
	\begin{eqnarray*}
		&&REV_2-REV  \\
		&=&\int_0^1 \hat{v}(q)(1-q)^{n-3}(-q^2(n-2)+2q(1-q)+q(1-q)(n-1)-nq(1-q)^2) dq\\
		&=&\int_0^1 \hat{v}(q)(1-q)^{n-3}(q+q^2-nq^3) dq
	\end{eqnarray*}
	Define $F(q)=(1-q)^{n-3}(q+q^2-nq^3)$
	Note that for $q<\frac{\sqrt{1+4n}-1}{2n}$, $F(q)>0$ and when $\frac{\sqrt{1+4n}-1}{2n}<q<1$, $F(q)<0$. Define $q^*=\frac{\sqrt{1+4n}-1}{2n}$, by the monotonicity of $\hat{v}(q)$, we have
	\begin{eqnarray*}
		\int_0^1 F(q)\hat{v}(q) dq&=&\int_0^{q^*} F(q)\hat{v}(q)+\int_{q^*}^1 F(q)\hat{v}(q)\\
		&\geq& \int_0^{q^*} \hat{v}(q^*)F(q)dq+ \int_{q^*}^1 \hat{v}(q^*)F(q) dq=\hat{v}(q^*)\int_0^1 F(q) dq
	\end{eqnarray*}
	
	So it is remain to show that $\int_0^1 F(q) \geq 0$, actually,
	$$\int_0^1 F(q) =\frac{1}{n}-\frac{1}{n+1}>0$$
	
	So the Lemma is proved. Noth that as $n\rightarrow \infty$ fake distribution $\hat{v} \rightarrow v(q)$ and the revenue difference approaching to 0.
\end{proof}

\section{Missing proofs from section 8}

\begin{proof}[Proof of Lemma \ref{coefficients}]
	Since the mechanism family is truthful, for any fake distribution $\hat{v}(\cdot)$, the allocation rule is monotone. Then $\frac{dR}{dq} \leq 0$ for any decreasing function $\hat{v}(\cdot)$, so for any $q\in [0,1]$
	\begin{eqnarray}
	\label{coefficients_equation}
	\frac{dR(q,\hat{v}(q),\hat{v}'(q))}{dq}=\frac{\partial R}{\partial q}+\frac{\partial R}{\partial \hat{v}} \hat{v}'(q)+\frac{\partial R}{\partial \hat{v}'} \hat{v}''(q) \leq 0
	\end{eqnarray}
	
	For any given $q,\hat{v}(q),\hat{v}'(q)$, if $\frac{\partial R}{\partial \hat{v}'} \neq 0$, set $\hat{v}''(q)$ to be a sufficiently large value (for positive $\frac{\partial R}{\partial \hat{v}'}$) or sufficiently small (for negative $\frac{\partial R}{\partial \hat{v}'}$), then the left side of (\ref{coefficients_equation}) is positive, contradiction.
	
	So $\frac{\partial R}{\partial \hat{v}'}=0$. Given this, for any given $q,\hat{v}(q)$, set $\hat{v}'(q)=0$, from (\ref{coefficients_equation}) we get $\frac{\partial R}{\partial q} \leq 0$. Set $\hat{v}'(q)$ sufficiently small, we get $\frac{\partial R}{\partial \hat{v}(q)} \geq 0$.
\end{proof}
Fix the other buyer's fake distribution, given quantile $q$, the interim allocation can be written as a function of $R$, and by Lemma \ref{coefficients}, $x^*(q)=x(R(q,\hat{v}(q)))$
\begin{proof}[Proof of Lemma \ref{virtual_efficient}]
	
	Consider a buyer with true type distribution $v(\cdot)$. Fix the other buyer's fake distribution, given quantile $q$, his interim allocation can be written as a function of $R$, and by Lemma \ref{coefficients}, $x^*(q)=x(R(q,\hat{v}(q)))$.
	
	By the definition and Lemma \ref{coefficients}, the buyer's utility is
	$$U=\int_0^1 x(R(q,\hat{v}(q)))(v(q)-\hat{v}(q)-q\hat{v}'(q))$$
	We regard this as a functional of $\hat{v}(\cdot)$ and define functional $F=x(R(q,\hat{v}(q)))(v(q)-\hat{v}(q)-q\hat{v}'(q))$. By
	Euler-Lagrange equation, we get
	$$ \frac{\partial F}{\partial \hat{v}}-\frac{d}{dx} \frac{\partial F}{\partial \hat{v}'}=0$$
	And we have
	
	$$\frac{\partial F}{\partial v}=\frac{\partial x}{\partial \hat{v}}(v-\hat{v}-q\hat{v}')-x$$
	
	$$\frac{\partial F}{\partial v'}=\frac{\partial x}{\partial \hat{v}'}(v-\hat{v}-q\hat{v}')-qx$$
	
	$$\frac{d}{dx} \frac{\partial F}{\partial \hat{v}'}=\frac{d \frac{\partial x}{\partial \hat{v}'}(v-\hat{v}-q\hat{v}')}{dq}-x-q\frac{dx}{dq}
	=-x-q\frac{\partial x}{\partial \hat{v}}v'-q\frac{\partial x}{\partial q}$$
	The last equality holds because $\frac{\partial x}{\partial y'}=\frac{\partial x}{\partial R}\frac{\partial R}{\partial \hat{v}'}=0$
	
	So
	\begin{eqnarray}
	\label{euler_equation}
	\frac{\partial x}{\partial \hat{v}}(v-\hat{v})+\frac{\partial x}{\partial q}=0
	\end{eqnarray}
	Note that the equation only related to variable $\hat{v}(q)$ and $q$, it is not a differential equation. So if the equation have no solution, or the solution $\hat{v}(q)$ is not monotone decreasing, then there is no equilibrium of fake distribution.
	
	If the equation has a decreasing solution $\hat{v}(q)$. Since $\frac{\partial x}{\partial \hat{v}}=\frac{\partial x}{\partial R}\frac{\partial R}{\partial \hat{v}}$,$\frac{\partial x}{\partial \hat{v}}=\frac{\partial x}{\partial q}\frac{\partial R}{\partial q}$, by definition of the allocation rule $x$, $\frac{\partial x}{\partial R}$ always non-negative. And by Lemma 8.3 we have $\frac{\partial x}{\partial \hat{v}} \geq 0$, $\frac{\partial x}{\partial q} \leq 0$. By this we get $v \geq \hat{v}$

	Note that the nature boundary condition $\frac{\partial F}{\partial \hat{v}'}|_{q=0,1}=0$
	is $x^*(1)=0$. it is automatically satisfied for symmetric equilibrium.
\end{proof}

\begin{proof}[Proof of Theorem \ref{less_than_second_price}]
	The revenue on equilibrium of fake distribution is
	$$REV=n\int_0^1 (1-q)^{n-1}(\hat{v}(q)+qv(q))dq=n \int_0^1 (1-q)^{n-1}d(q\hat{v}(q))=\int_0^1 (n-1)(1-q)^{n-2}\hat{v}(q) dq$$
	
	Substitute $\hat{v}$ for $v$ we get the revenue of second-price auction $\int_0^1 (n-1)(1-q)^{n-2}v(q) dq$
	
	By Lemma \ref{virtual_efficient}, we get $REV \leq REV_2$.
	
\end{proof}
\section{some other mechanism families}
In this section we discuss two special mechanism families.

\begin{definition}(Mechanism family with target distribution $\mathbf{v}(\cdot)$)
	
	The mechanism family $\{M^{\mathbf{f}(\cdot)}|\mathbf{f}(\cdot) \in \mathcal{F}^n\}$, set all the allocation rule to 0 in all mechanisms, except one special mechanism $M^{\mathbf{f}(\cdot)}$ such that $f_i(q)+qf_i'(q)=v_i(q)$ for any $i$.
	
	Define $M^{\mathbf{f}(\cdot)}$ as following:
	
	Allocate the item to the buyer with the largest virtual value $f(q)+qf'(q)$. The payment rule is according to (\ref{payment}) (give slightly disturb so that buyers can get positive utility).
	
\end{definition}

\begin{lemma}
	Given buyers' true valuation distributions $\mathbf{v}(\cdot)$, the mechanism family with target distribution $\mathbf{v}(\cdot)$ can extract almost all the social surplus, i.e., $$REV=\int_{\cdot} \max_i v_i(q_i) d\mathbf{q}$$
\end{lemma}
\begin{proof}
	By definition, the only equilibrium for buyers are $\hat{v}_i(q)=v_i(q)$, for any $i$, otherwise they both will get 0 utility.
	
	In mechanism $M^\mathbf{f}(\cdot)$, from definition we know that, fix the quantile profile $\mathbf{q}$, the revenue is the maximum virtual value of $f_i$, that is $\max v_i(q)$.
\end{proof}

However, this mechanism only performs well when buyers have the exact true type distribution. It gets 0 revenue at all other cases. Since the prior distributions are private information, the seller can not design this kind of mechanism family in advance.

In next mechanism we will show that the revenue gap between normal setting and fake distribution setting can be infinite.
\begin{definition}(Mechanism family with respect to quantile)
	The mechanism allocate the item to the buyer whose corresponding quantile of the bid $b$ ,i.e., $1-F(b)$, is minimum. The payment rule is according to (\ref{payment}).
\end{definition}

It is a truthful mechanism family.

\begin{lemma}
	For mechanism family with respect to quantile, the revenue gap between normal setting and fake distribution setting can be infinite.
\end{lemma}

\begin{proof}
	Consider symmetric environment, in normal setting, i.e., $\hat{v}(q)=v(q)$, by payment identity the revenue is equal to the revenue of second-price auction with prior distribution $v(q)$. In fake distribution setting, since buyer's interim allocation rule is not related to $\hat{v}(q)$ ($x^*(q)=1-q$), the buyer can report a fake distribution $\hat{v}(q)=0$, so that the total expect payment is 0. Note that by definition the seller always guarantees  a non-negative revenue, so this fake distribution is an equilibrium of fake distribution, of which the revenue is 0.
	
	So the revenue gap can be infinite.
\end{proof}






\end{document}